\documentclass[letterpaper, 10 pt, conference]{ieeeconf}  % Comment this line out if you need a4paper

\IEEEoverridecommandlockouts                              % This command is only needed if 
\usepackage[utf8]{inputenc}
\usepackage[T1]{fontenc}
\usepackage{amsmath}
\usepackage{amssymb}
\usepackage{graphicx}
\usepackage[english]{babel}
\usepackage{bm}
\usepackage{todonotes}
\usepackage{subcaption}
\usepackage{siunitx}

\usepackage[backend=biber,
style=ieee, % bis 3 Anfangsbuchstaben, + bei mehr Autoren, Jahresangabe zweistellig
isbn=false,
doi=false,
date=year,
]{biblatex}
\usepackage{tikz}
\usetikzlibrary{arrows, positioning, calc, fit}

\usepackage{environ}
\newtheorem{definition}{Definition}
\newtheorem{remark}{Remark}
\newtheorem{thm}{Theorem}
\newtheorem{lem}{Lemma}
\newtheorem{ass}{Assumption}
\newtheorem{prop}{Proposition}

\makeatletter
\newsavebox{\measure@tikzpicture}
\NewEnviron{scaletikzpicturetowidth}[1]{%
	\def\tikz@width{#1}%
	\begin{lrbox}{\measure@tikzpicture}%
		\BODY
	\end{lrbox}%
	\pgfmathparse{#1/\wd\measure@tikzpicture}%
	\BODY
}
\makeatother

\newcommand{\hinf}{$\mathcal{H}_\infty$}

\newcommand{\N}{\mathbb{N}}
\newcommand{\R}{\mathbb{R}}

\newcommand{\signpow}[1]{\left\lceil {#1} \right\rfloor}
\newcommand{\qhnorm}[4]{\left\| {#1}\right\|_{{#2},{#3}}^{#4}}

\newcommand{\gradient}[2]{\frac{\partial {#1}\left({#2}\right)}{\partial {#2}}}
\newcommand{\lp}[1][p]{\mathcal{L}_{{#1}}}
\newcommand{\lph}[1][p]{\mathcal{L}_{{#1}\mathrm{h}}}
\newcommand{\lpnorm}[2][p]{\| #2 \|_{\lp[#1]}}
\newcommand{\lphnorm}[3][p]{\| #2 \|_{#3,\lph[#1]}}

\title{%\LARGE \bf
	Minimizing the Homogeneous $\lp[2]$-Gain of\\Homogeneous Differentiators
}

\author{Benjamin Vo\ss$^{1}$, Jaime A. Moreno$^{2}$ and Johann Reger$^{1\ast}$% <-this % stops a space
	\thanks{$^{1}$Control Engineering Group, Technische
		Universit\"at Ilmenau, P.O. Box 10 05 65, D-98684, Ilmenau, Germany}%
		\thanks{$^{2}$Eléctrica y Computatión, Instituto de Ingeniería, Universidad National Autónoma de México. Ciudad de México 04510, Mexico}
	\thanks{$^{\ast}$Corresponding author: {\tt\small johann.reger@tu-ilmenau.de}}
}

\begin{document}

	\maketitle
	\thispagestyle{empty}
	\pagestyle{empty}

	%%%%%%%%%%%%%%%%%%%%%%%%%%%%%%%%%%%%%%%%%%%%%%%%%%%%%%%%%%%%%%%%%%%%%%%%%%%%%%%%
	\begin{abstract}
	The differentiation of noisy signals using the family of homogeneous differentiators is considered.	It includes the high-gain (linear) as well as robust exact (discontinuous) differentiator. 
	To characterize the effect of noise and disturbance on the differentiation estimation error, the generalized, homogeneous $\bm{\lp[2]}$-gain is utilized. Analog to the classical $\bm{{\lp}}$-gain, it is not defined for the discontinuous case w.r.t. disturbances acting on the last channel. Thus, only continuous differentiators are addressed. 
	The gain is estimated using a differential dissipation inequality, where a scaled Lyapunov function acts as storage function for the homogeneous $\bm{\lp[2]}$ supply rate.
	The fixed differentiator gains are scaled with a gain-scaling parameter similar to the high-gain differentiator. 
	This paper shows the existence of an optimal scaling which (locally) minimizes the homogeneous $\bm{\lp[2]}$-gain estimate and provides a procedure to obtain it. 
	Differentiators of dimension two are considered and the results are illustrated via numerical evaluation and a simulation example.
		
	\end{abstract}

	%%%%%%%%%%%%%%%%%%%%%%%%%%%%%%%%%%%%%%%%%%%%%%%%%%%%%%%%%%%%%%%%%%%%%%%%%%%%%%%%

	\section{Introduction}

	Differentiation of noisy signals remains a challenging task, where mainly two approaches are used to tackle the problem: the linear high-gain observer acting as a differentiator on the one hand~\cite{Khalil2002,Atassi2000,Vasiljevic2008} and the discontinuous robust exact differentiator on the other hand~\cite{Levant1998,Levant2001,Levant2003}. Both these approaches are special cases of the generalized family of homogeneous differentiators of degree $d$ (see e.g. \cite{CruzZavala2016,Yang2004}), where $d = 0$ corresponds to the linear case and $d = -1$ to the discontinuous case.
	
	It is of natural interest to characterize the influence of noise and disturbances on the differentiation estimation error, e.g. for appropriate parameter tuning or robust stability of the closed loop. 
	Vasiljevic and Khalil ~\cite{Vasiljevic2008} make use of the $\lp[\infty]$-gain and derive an optimal choice of the high-gain parameter. This minimizes an upper bound of the $\lp[\infty]$-gain, given the knowledge of the respective noise and disturbance bounds.
	For the discontinuous, second-order case, i.e. the super-twisting differentiator, Seeber~\cite{Seeber2023} just recently published non-conservative error bounds in a similar setting. Furthermore, a trade-off between speed of convergence and differentiation error are discussed to ease adequate tuning. 
	
	The goal of this article is to find an optimal gain-scaling~$L$ (similar to the high-gain setup) for the homogeneous differentiator~\cite{CruzZavala2016}. For this purpose, we need to consider weighted homogeneous systems with inputs and outputs, i.e. homogeneous input-output mappings. When introducing weights not only to the state, input and output but also to the time variable, applying a dilated input (also w.r.t. time) can lead to a dilated output. This yields the concept of homogeneous input-output maps recently introduced by Zhang~\cite{Zhang2022}. 
	It turns out that the classical $\lp$-gain~\cite{Schaft2017} is only a local property for these systems, i.e. it is not dilation-invariant (apart from special cases). Thus, Zhang proposes a generalized, homogeneous  $\lp$-gain~\cite{Zhang2022}  ($\lph$-gain) with homogeneity degree zero and therefore global validity. Being a generalization, the classical $\lp$-gain is recovered for linear systems (i.e. $d = 0$).
	
	Therefore, we consider optimality with respect to the $\lph$-gain of the homogeneous continuous differentiation error dynamics of degree $d\in(-1,1)$. The discontinuous case ($d=-1$) is explicitly excluded here, since the $\lph$-gain is not defined for disturbances acting on the last channel, i.e. signals with non-vanishing $n$-th derivative in this case. 
	Similar to the high-gain setup~\cite{Vasiljevic2008}, we consider scaled gains $k_i = \alpha_i L^i$, $i=1,\ldots,n$, where the values of $\alpha_i$ are fixed and gain-scaling $L>0$ is the optimization variable. 
	In the present paper, we restrict ourselves to the homogeneous $\lp[2]$-gain with a clear physical interpretation of the input-output ratio in terms of energy (for the linear case). Further, we choose $n=2$ for illustration purposes.
	We show that the $\lph[2]$-gain estimate grows linearly with gain-scaling $L$ if noise is present  and is inversely proportional to a power of $L$ regarding the disturbance.
	Thus, we prove the existence of a global optimum and propose a procedure to find an optimal scaling $L$ that minimizes the effect of bounded noise and disturbance on the differentiation estimation error in terms of the estimated $\lph[2]$-gain. 
		
	The calculation of the $\lp[2]$-gain for linear systems, where it coincides with the \hinf-norm of the corresponding transfer function, can be done in the frequency domain (see e.g.~\cite{Doyle1990}). The more general nonlinear case is extensively studied e.g. by v. d. Schaft~\cite{Schaft2017}, and generalized to the homogeneous case by Zhang~\cite{Zhang2022}, where the gain is calculated based on the solution to a homogeneous (differential) dissipation inequality (DI).
	This defines a partial differential inequality (PDI) and is hard to solve in general (apart from the linear case, where it reduces to a Riccati inequality~\cite{Doyle1990}).  From asymptotic stability of the error dynamics and the converse Lyapunov theorem, however, the existence of a (homogeneous) Lyapunov function is guaranteed~\cite{Bacciotti2005,Bhat2005}. Knowing that the DI needs to hold for all permissible inputs $u$ (including $u\equiv 0$), any Lyapunov function qualifies as a candidate storage function, i.e. a solution to the PDI. Thus, we utilize the homogeneous Lyapunov function of Cruz-Zavala and Moreno~\cite{CruzZavala2016} at the expense of estimating only an upper bound of the $\lph[2]$-gain.
	 
	In terms of a specific parameter set, we calculate the $\lph[2]$-gain estimates for various homogeneity degrees  and relate the estimate to the actual $\lp[2]$-gain for the linear case. Focusing on a specific homogeneity degree, we provide a simulation example and observe that optimal gain-scaling $L^*$ yields a reasonable trade-off between noise and disturbance affecting the differentiation estimation error. 

	Since the concept of homogeneous input-output maps as well as the homogeneous $\lp$-gain have been introduced just recently, we elaborate on that in Sections~\ref{sec:preliminaries} and~\ref{sec:motivation}. Furthermore, we recall necessary definitions of weighted homogeneous systems and classical $\lp$-gain and show, by means of a simple example, that the latter is not suitable for homogeneous systems.
	Then, the problem statement is presented in Section~\ref{sec:problem_statement} together with the Lyapunov and storage function. Section~\ref{sec:l2h-gain} covers the $\lph[2]$-gain estimation for the homogeneous differentiator and its minimization w.r.t. gain-scaling $L$. Moreover, a numerical evaluation for an exemplary parameter set and the simulation example are part of the section. Finally, conclusions are drawn in Section~\ref{sec:conclusions}.

\section{Preliminaries}\label{sec:preliminaries}
\noindent
	In this article, we consider systems
\begin{equation}\label{eq:system}
			\Sigma: \left\lbrace
	\begin{aligned}
		\dot{x}(t) &= f(x(t),u(t)),\quad x(t_0) = x_0,\\
		y(t) &= h(x(t),u(t))
	\end{aligned}
			\right.
\end{equation}
with input $u(t)\in\R^{n_u}$, state $x(t)\in\R^{n}$ and output $y(t)\in\R^{n_y}$ vectors of dimensions $n_u,n,n_y\in\N$. Both the vector field $f(x,u)$ and function $h(x,u)$ are continuous in $x$ and~$u$ with $u$ being measurable and essentially bounded.
Thus, a solution to~\eqref{eq:system} exists~\cite{Bacciotti2005} and $\Sigma$ defines an input-output mapping $G_{x_0}$ for every $x_0$. That is, solving~\eqref{eq:system} with a given (admissible) input $u(\cdot)$ for the trajectory $x(\cdot)$ leads to the corresponding output $y(\cdot)$.

We briefly recapitulate the main concepts used in the article and show that their straight-forward conjunction is questionable, since the classical $\lp$-gain is a local property for nonlinear, homogeneous systems with inputs and outputs.

\subsection{Classical $\lp$-Stability and $\lp$-Gain}\label{sec:lp_gain}
\noindent
We consider $n$-dimensional Lebesgue-measurable signals $f:[0,\infty)\to\R^n$ with $f\in\lp^n$, i.e.
%\begin{equation*}
	$\int_{0}^{\infty} \| f(t)\|^p \mathrm{d}t < \infty$, $p\geq 1$~\cite{Schaft2017},
%\end{equation*}
where $\|\cdot\|$ is any norm in $\R^n$.  This leads us to the definition of the $\lp$-norm as~\cite{Schaft2017}
\begin{equation*}
	\| f \|_{\lp} = \left( \int_{0}^{\infty} \| f(t)\|^p \mathrm{d}t  \right)^{\frac{1}{p}},\quad p\geq 1
\end{equation*}

In a similar manner, the extended $\lp^n$-space, i.e. $\mathcal{L}_{p\mathrm{e}}^n$, can be defined for truncated signals $f_\mathrm{T}$ with $f_\mathrm{T}(t) = f(t),~t\in[0,T)$ and $f_\mathrm{T}(t) = 0,~t\geq T$.

An input-output mapping $G:\mathcal{L}_{p\mathrm{e}}^{n_u}\to \mathcal{L}_{p\mathrm{e}}^{n_y}$ is called $\lp$-stable, if~\cite{Schaft2017}
\begin{equation*}
	u\in\lp^{n_u} \Longrightarrow y = G(u) \in \lp^{n_y}.
\end{equation*}
The mapping has finite $\lp$-gain, if there exist finite constants $\gamma_p$ and $b_p$ such that%~\cite{Schaft2017}
\begin{equation}\label{eq:lp_gain}
	\| G(u) \|_{\mathcal{L}_p}\leq \gamma_p \|u\| _{\mathcal{L}_p} + b_p,
\end{equation}
which implies $\lp$-stability of $G$~\cite{Schaft2017}.
The $\lp$-gain of $G$ is defined as $\gamma_p(G) = \inf\{\gamma_p~|~\exists \,b_p\text{ s.t. \eqref{eq:lp_gain} holds}\}$~\cite{Schaft2017}.

Let us consider the $\lp$-gain of state-space representation~$\Sigma$. For this purpose,  we introduce the continuously differentiable storage function $V:\R^n\to[0,\infty)$ and recall the $\lp$ supply rate $s_p:\R^{n_u}\times \R^{n_y}\to\R$ with 
		\begin{equation}\label{eq:l2_supply_rate}
	s_p(u,y) = \gamma^p_p\|u\|^p- \|y\|^p,\quad \gamma\geq 0.
\end{equation} 
System~\eqref{eq:system} has $\lp$-gain $\leq \gamma_p$ if the differential dissipation inequality (DDI)~\cite{Schaft2017}
\begin{equation}\label{eq:DDI}
	\gradient{V}{x}f(x,u) \leq s_p(u,h(x,u))~\forall x,u
\end{equation}
holds. That is, $\Sigma$ is dissipative w.r.t. supply rate $s_p$.
The $\lp$-gain of system~\eqref{eq:system} is defined as $\gamma_p(\Sigma) = \inf\{\gamma_p~|~\Sigma\text{ has $\lp$-gain }\leq \gamma_p\}$~\cite{Schaft2017}.

These classical results can be extended to the well-known case $p = \infty$. Another prominent case is $p=2$, since the physical meaning is obvious and the $\lp[2]$-gain represents the input-output ratio in terms of energy.

After a short introduction into weighted homogeneity and its expansion to systems with inputs and outputs, we consider the homogeneous differentiator and come back to the $\lp$-gain of the corresponding error dynamics.

%In order to calculate $\gamma$ for a given map, we assume the existence of 

%\begin{equation*}
%f_\mathrm{T}(t) =\begin{cases}
%	f(t), &0\leq t< T \\ 0, &t\geq T
%\end{cases}
%\end{equation*}
	\subsection{Weighted Homogeneous Systems with Inputs and  Outputs}
	\noindent
	Following Baciotti and Rosier, we define:
	
	\begin{definition}[Weighted Homogeneous System~{\cite[Chap.~5]{Bacciotti2005}}]\label{def:homogenous_system}
		The \emph{weight vector} associated with the state vector $x = (x_1,\ldots,x_n)^\top$ is denoted by $r = (r_{1},\ldots,r_{n})^\top$ with positive weights $r_{i}>0$. The respective weight vectors for $u$ and $y$ are called $r_u$ and $r_y$. 
		The corresponding \emph{dilation} operator is defined as $\Delta_\kappa^{r}(x) := (\kappa^{r_{1}}x_1,\ldots,\kappa^{r_{n}}x_n)^\top, \kappa>0$. 
		The system~\eqref{eq:system} is called \emph{homogeneous  of degree} $d\in(-\min_i r_{i},\infty)$ if there exist positive weights $r_{i}, r_{u_i},r_{y_i}, i= 1,2,\ldots$ s.t. $\forall u,x$ and $\forall \kappa > 0$ the following holds:
		\begin{align*}
			f_i(\Delta_\kappa^{r}(x),\Delta_\kappa^{r_u}(u)) &= \kappa^{d+r_{i}}f_i(x,u), &\forall i =& 1,\ldots,n,\\
			h_j (\Delta_\kappa^{r}(x),\Delta_\kappa^{r_u}(u)) &= \kappa^{r_{y_j}}h_j(x,u), &\forall j =& 1,\ldots,n_y.
		\end{align*} 
	\end{definition} 
	In order to describe homogeneity of an input-output map~$G$, e.g. $G_{x_0}$ principally defined by \eqref{eq:system}, it is necessary to introduce the weight $r_t=-d$ associated with time $t$. This leads to
	\begin{equation*}
		f_i(\Delta_\kappa^{r}(x),\Delta_\kappa^{r_u}(u)) = \kappa^{d+r_{i}}f_i(x,u) = \frac{\kappa^{r_{i}}\mathsf{d}x_i}{\kappa^{r_t}\mathsf{d}t}
	\end{equation*}
	and means that the trajectory of~\eqref{eq:system} is homogeneous w.r.t. input and time~\cite{Zhang2022}. Thus, a dilated input leads to a dilated state transition, which in turn implies a dilated output and leads to the following definition. 
	\begin{definition}[Homogeneous Input-Output Map~{\cite[Def.~2]{Zhang2022}}]
		The causal and time-invariant input-output map $G$ is termed \emph{$r$-homogeneous of degree} $d=-r_t\in\R$, if
		\begin{equation*}
			G(\Delta_\kappa^{r_u}(u(\kappa^{-r_t}\cdot))) = \Delta_\kappa^{r_y}(y(\kappa^{-r_t}\cdot)),\quad \forall\kappa>0
		\end{equation*} 
		holds for each admissible input $u$ and output $y = G(u)$.
	\end{definition}
	This definition is essential to investigate the $\lp$-gain of homogeneous systems.

		Aside from that, we make use of the $r$-homogeneous $q$-norm defined as~\cite{Bacciotti2005}
\begin{equation}\label{eq:qhnorm}
	\qhnorm{x}{r_x}{q}{}= \left(\sum_{i = 1}^{n}|x_i|^{\frac{q}{r_{x_i}}}\right) ^{\frac{1}{q}},\quad \forall x\in\R^{n},~q\geq 1
\end{equation}
	and unit sphere $\mathcal{S}_{r_x,q} = \left\{\left.x\in\R^{n}~\right|~\qhnorm{x}{r_x}{q}{} = 1\right\}$.	
	Furthermore, consider the following essential property.
	\begin{lem}[Dominating homogeneous function \cite{Andrieu2008,Bhat2005,Hestenes1966}]\label{lem:homog_domination}
		Let $\mu:\R^n\to\R_{\geq 0}$ (i.e. $\mu(x)\geq 0~\forall x\in\R^n$) and $\eta:\R^n\to\R$ be two continuous homogeneous functions with weights $ r= (r_1,\ldots,r_n)$ and degrees $d$, such that
		\begin{equation*}
			\left\lbrace x\in\R^n\backslash\{0\}:\mu(x) = 0\right\rbrace \subseteq
			\left\lbrace x\in\R^n\backslash\{0\}: \eta(x) < 0\right\rbrace.
		\end{equation*}
		Then, there exists $\lambda^*\in\R$ such that
		\begin{equation*}
			\eta(x) - \lambda \mu(x) < -c\|x\|_{r,p}^{d}\quad
			\forall x\in\R^n\backslash \{0\},~\forall \lambda\geq\lambda^*
		\end{equation*}
		holds for some $c>0$.
	\end{lem}

	\subsection{Homogeneous Arbitrary Order Differentiator}
	We consider the homogeneous differentiator given by~\eqref{eq:system} with~\cite{CruzZavala2016}
		\begin{subequations}\label{eq:hnod}
		\begin{align}
			f_i(x,f_\mathrm{n}) &= -k_i \signpow{x_1 - f_\mathrm{n} } ^{\frac{r_{i+1}}{r_1}} + x_{i+1},~ i = 1,\ldots,n-1 \nonumber\\
			f_{n} (x,f_\mathrm{n}) &= -k_{n} \signpow{x_1 - f_\mathrm{n} } ^{\frac{r_{n+1}}{r_1}},\\
			h(x,f_\mathrm{n}) &= x_{n}-f_0^{(n-1)}
		\end{align}
	\end{subequations}
	where $\signpow{\cdot}^p = \mathrm{sign}(\cdot)|\cdot|^p,~p\in\R$ denotes the sign-preserving power. The  Lebesgue-measurable function $f_\mathrm{n}(t) = f_0(t) - \nu(t)$ consists of $n$-times differentiable base signal  $f_0$ to be differentiated and bounded noise $|\nu| \leq N$ with $N\geq 0$. The weights are assigned as $r_i = r_{i+1}-d = r_{n} -(n-i)d$, $i=1,\ldots,n+1$, where we fix $r_{n} = 1$ and allow $d\in(-1,1)$ and the gains are $k_i>0$. We formally introduce the output function $h$ as $(n-1)$-th error in the differentiation of base signal $f_0$.

	In order to utilize the Lyapunov function of~\cite{CruzZavala2016,CruzZavala2019}, consider the scaled error dynamics
	\begin{subequations}\label{eq:hnod_z}
		\begin{align}
			\dot{z}_i &= -\tilde{k}_i \left(\signpow{z_1+\nu}^{\frac{r_{i+1}}{r_1}}- z_{i+1}\right), ~ i = 1,\ldots,n-1 \nonumber\\
			\dot{z}_n &= -\tilde{k}_n \signpow{z_1 + \nu}^{\frac{r_{n+1}}{r_1}}+ \tilde{\delta}, \\
			y &= \tilde k_1 z_n,
		\end{align}
	\end{subequations}
	where $
%	\begin{equation*}
		z_i =  \frac{x_i- f_0^{(i-1)}}{k_{i-1}},  
		\tilde{k}_i = \frac{k_i}{k_{i-1}}$ for $i=1,\ldots,n$ 
%	\end{equation*}
	and $k_0 = 1$.
	The disturbance $\delta = -f_0^{(n)}$ is assumed to be bounded by $D\geq 0$, i.e. $|\delta| \leq D$, and scaled with $\displaystyle\tilde{\delta} = \frac{\delta}{k_1}$. Although $\delta\in\lph$ (the homogeneous $\lp$-space, see Sec.~\ref{sec:lph_gain}) is sufficient for the calculations in this article, we require boundedness of $\delta$ to ensure ultimate uniform  boundedness of the differentiation estimation error for $d\in(-1,0]$~\cite{CruzZavala2016}.

%	\begin{remark}[Certain homogeneity degrees $d$]
%		In the linear case, i.e. for $d=0$, the dynamics~\eqref{eq:hnod} coincide with the high-gain differentiator of~\cite{Vasiljevic2008}, where the gains are chosen as $k_i = \frac{\alpha_i}{\varepsilon^i}$. In contrast, Levant's robust exact differentiator\todo{cite} is recovered for $d=-1$. This case is explicitly excluded here, since the $\lp$-gain stops being defined for the error dynamics~\eqref{eq:hnod_z} (they are either unstable or not affected by disturbance $\delta$).
%	\end{remark}
	
%	\begin{remark}[Homogeneity with respect to inputs and output]\label{remark:hnod_homogeneous}
%\todo[inline]{\eqref{eq:hnod_z} defines a homogeneous input-output mapping }
%	\end{remark}
	
	Note that with weight $r_t = -d$ the error dynamics define a homogeneous input-output map, where $u = (\nu,~\delta)^\top$ with weights $r_u = (r_\nu,~r_\delta) =(r_1,~r_{n+1})$.

	\section{Classical $\lp$-Gain versus Homogeneity}\label{sec:motivation}
	For simplicity of exposition, we reduce the following derivations to the two-dimensional case. 
	By means of the homogeneous differentiator, we show that the classical $\lp$-gain is a local property only, i.e. it lacks invariance with respect to homogeneous dilation. This leads to the introduction of a homogeneous $\lp$-gain. 
%	which we (over-)estimate using a scaled Lyapunov function as storage function for the corresponding homogeneous $\lp$ supply rate.
	\subsection{Second-order Homogeneous Differentiator}\label{sec:h1od_z}
	The error dynamics~\eqref{eq:hnod_z} with $n=2$ read
		\begin{subequations}\label{eq:h1od_z}
		\begin{align}
			\dot{z}_1 &= -\tilde{k}_1 \left(\signpow{z_1+\nu}^{\frac{1}{1-d}} - z_2\right), &z_1(0)& = z_{1,0}\\
			\dot{z}_2 &= -\tilde{k}_2 \signpow{z_1 + \nu}^{\frac{1+d}{1-d}} + \tilde{\delta}, &z_2(0)& = z_{2,0}\\
			y &= \tilde k_1 z_2,
		\end{align}
	\end{subequations}
%	where 
%	\begin{equation*}
%		z_1 = x_1 - f_0,~z_2 = \frac{x_2- f_0^{(1)}}{k_1}.\quad \text{and}\quad
%		\tilde{k}_1 = k_1,~\tilde{k}_2 = \frac{k_2}{k_1},
%	\end{equation*}
%	the disturbance $\delta = -f_0^{(2)}$ is bounded by $D\geq 0$, i.e. $|\delta| \leq D$, and scaled with $\displaystyle\tilde{\delta} = \frac{\delta}{k_1}$.
	where the respective weights of inputs $u = (\nu,~\delta)^\top$, output and states are given by
	\begin{align*}
		r_u &= (r_\nu,~r_\delta) = (1-d,~1+d), &r_y &=  1,\\
		r_z &= (r_1,~r_2) = (1-d,~1). & &		
	\end{align*}

	\subsection{Limitations of the $\lp$-Gain for Homogeneous Systems}\label{sec:limitations}
	To investigate the suitability of the classical $\lp$-gain for homogeneous systems, consider the special case of~\eqref{eq:h1od_z} where no noise is present, i.e. $\nu\equiv 0$.
	For zero initial conditions, the nonzero input $u_1 = (0,~\delta_1)$ with $\delta_1\in\lp$ yields the corresponding output~$y_1$. This results in the ratio
	\begin{equation*}
		\Gamma(u_1,y_1) = \frac{\lpnorm{y_1}}{\lpnorm{u_1}}.
	\end{equation*}
	Recall that~\eqref{eq:h1od_z} defines a homogeneous input-output mapping, i.e. applying the dilated input $u_2 (\cdot) = \Delta_\kappa^{r_u}(u_1(\kappa^{-r_t}\cdot))$ yields the dilated output $y_2(\cdot) = \Delta_\kappa^{r_y}(y(\kappa^{-r_t}\cdot))$. The corresponding $\lp$-norms read
	\begin{align*}
		\lpnorm{u_2} &= \left( \int_{0}^{\infty} | \delta_2(t)|^p \mathrm{d}t  \right)^{\frac{1}{p}} = \left( \int_{0}^{\infty} | \kappa^r_\delta\, \delta_1(\kappa^{-r_t}t)|^p \mathrm{d}t  \right)^{\frac{1}{p}}\\
			& = \kappa^{r_\delta + \frac{r_t}{p}}\left( \int_{0}^{\infty} | \delta_2(\tilde t)|^p \mathrm{d}\tilde t  \right)^{\frac{1}{p}} = 	\kappa^{r_\delta + \frac{r_t}{p}}\lpnorm{u_1},\\
		\lpnorm{y_2} &= \kappa^{r_y + \frac{r_t}{p}}\lpnorm{y_1}
	\end{align*}
	and lead to the ratio
		\begin{equation*}
		\Gamma(u_2,y_2) = \frac{\lpnorm{y_2}}{\lpnorm{u_2}} = \kappa^{r_y-r_\delta} \Gamma(u_1,y_1)  = \kappa^{-d} \Gamma(u_1,y_1).
	\end{equation*}
	This ratio directly relates to $\gamma_p$ of~\eqref{eq:lp_gain}~\cite{Zhang2022} and is only constant for $d=0$, i.e. the linear case. If $d<0$, it grows unbounded for $\kappa\to\infty$ and if $d>0$ this happens for $\kappa\to 0$. An illustration is depicted in Fig.~\ref{fig:gamma_quotient} of Section~\ref{sec:numerics}  for an exemplary parameter choice. We conclude that the classical $\lp$-gain is not suitable for homogeneous systems.

	\subsection{Homogeneous $\lp$-Gain}\label{sec:lph_gain}
	In order to define a dilation-invariant and global $\lp$-like gain, we follow a similar path to Sec.~\ref{sec:lp_gain}, where the measurable signals $f:[0,\infty)\to\R^n$ exhibit the weight vector $r_f$.
	Using the homogeneous $q$-norm~\eqref{eq:qhnorm}, define the homogeneous $\lp$-norm (i.e. the $\lph$-norm) as~\cite{Zhang2022}
	\begin{equation}\label{eq:lphnorm}
		\lphnorm{f}{r_f} = \left( \int_{0}^{\infty} \qhnorm{f(t)}{r_f}{q}{p} \mathrm{d}t  \right)^{\frac{1}{p}},\quad q\geq 1,~p\geq 1,
	\end{equation}
	provided the right-hand side exists, i.e. $f\in\lph^n$ (the homogeneous $\lp^n$-space)~\cite{Zhang2022}.
	Equipped with the $\lph$-norm and $\lph^n$-space, the concept of $\lph$-stability and $\lph$-gain can be defined analogously to Sec.~\ref{sec:lp_gain} for homogeneous input-output maps $G_\mathrm{h}$ (see~\cite{Zhang2022}).  Similar to~\eqref{eq:lp_gain},  $G_\mathrm{h}$ has finite $\lph$-gain if there exist constants $\gamma_{\mathrm{h}p}\geq 0$, $b_{\mathrm{h}p}\geq 0$
	\begin{equation}\label{eq:lph_gain}
		\| G_\mathrm{h}(u) \|_{\lph}\leq \gamma_{\mathrm{h}p}\|u\| _{\lph} + b_{\mathrm{h}p},~p\geq 1
	\end{equation}
	for $u\in\lph$ and the $\lph$-gain of $G_\mathrm{h}$ is defined as  $\gamma_{\mathrm{h}p}(G_\mathrm{h}) = \inf\{\gamma_{\mathrm{h}p}~|~\exists \,b_{\mathrm{h}p}\text{ s.t. \eqref{eq:lph_gain} holds}\}$~\cite{Zhang2022}.
	
	Let us consider the $\lph$-gain of the $r$-homogeneous state-space representation $\Sigma$ of degree $d = -r_t$ (see~\eqref{eq:system} and Def.~\ref{def:homogenous_system}).
	
	Given the $r$-homogeneous continuously differentiable storage function $V_\mathrm{h}:\R^n\to[0,\infty)$ of degree $p-d>0$, the system~$\Sigma$ has $\lph$-gain $\leq \gamma_{\mathrm{h}p}$ if it is dissipative w.r.t. the homogeneous $\lp$ supply rate~\cite{Zhang2022}
		\begin{equation}\label{eq:l2h_supply_rate}
		s_{\mathrm{h}p}(u,y) = \gamma_{\mathrm{h}p}^p\qhnorm{u}{r_u}{q}{p} - \qhnorm{y}{r_y}{q}{p}.
	\end{equation} 
	That is, defining the $r$-homogeneous of degree $p$ value-function
	\begin{equation}\label{eq:DI_value_function}
		\mathcal{J}(x,u,y;\gamma_{\mathrm{h}p}) = 	\gradient{V_\mathrm{h}}{x} f(x,u) + \qhnorm{y}{r_y}{q}{p}  - \gamma_{\mathrm{h}p}^p\qhnorm{u}{r_u}{q}{p},
	\end{equation}
	the homogeneous DDI (hDDI)~\cite{Zhang2022}
	\begin{equation}\label{eq:DI}
	\mathcal{J}(x,u,y;\gamma_{\mathrm{h}p}) <0%\leq -\epsilon \qhnorm{z}{r_z}{q}{p}
	\end{equation}
	holds for all $u\in\lph^{n_u}$ and $y\in\lph^{n_y}$. Finally, the $\lph$-gain of $\Sigma$ is defined as $\gamma_{\mathrm{h}p}(\Sigma) = \inf\{\gamma_{\mathrm{h}p}~|~\Sigma\text{ has $\lph$-gain }\leq \gamma_{\mathrm{h}p}\}$~\cite{Zhang2022}.\\

	\subsection{Example Revisited}
	Consider the example discussed in Sec.~\ref{sec:limitations} and denote the ratio of homogeneous $\lp$-norms by
	\begin{equation*}
		\Gamma_\mathrm{h}(u_1,y_1)  =  \frac{\lphnorm{y_1}{r_y}}{\lphnorm{u_1}{r_u}}.
	\end{equation*}
	A straightforward calculation of the dilated in- and output's $\lph$-norms (analogous to the $\lp$-norms) yields
	\begin{align*}
		\lphnorm{u_2}{r_u} &= \kappa^{1+\frac{r_t}{p}}\lphnorm{u_1}{r_u},\\
		\lphnorm{y_2}{r_y} &= \kappa^{1+\frac{r_t}{p}}\lphnorm{y_1}{r_y}	
	\end{align*}
	resulting in the ratio
	\begin{equation*}
			\Gamma_\mathrm{h}(u_2,y_2)  =  \frac{\lphnorm{y_2}{r_y}}{\lphnorm{u_2}{r_u}} =  \kappa^{0} \frac{\lphnorm{y_1}{r_y}}{\lphnorm{u_1}{r_u}} = \Gamma_\mathrm{h}(u_1,y_1),
	\end{equation*}
	which is constant $\forall\,\kappa>0$ and thus  dilation-invariant.
	We conclude that the homogeneous $\lp$-gain is suitable for homogeneous systems and for more details refer the interested reader to Zhang~\cite{Zhang2022}.
	
	\section{Problem Statement}\label{sec:problem_statement}
		We try to minimize the homogeneous $\lp[2]$-gain from input~$u$ to output~$y$ by appropriate choice of the gains $k_1,k_2$. In the spirit of the high-gain observer (e.g.~\cite{Atassi2000,Vasiljevic2008}), we make use of scaled gains
	\begin{equation}\label{eq:gains_scaled}
		k_1 = \alpha_1 L \quad \text{and}\quad k_2 = \alpha_2L^2
	\end{equation}
	with fixed $\alpha = (\alpha_1,~\alpha_2)^\top$ but variable gain scaling $L>0$.  This reduces the degrees of freedom significantly, however, simplifies a generalization to the arbitrary-order differentiator.

	In general DDI~\eqref{eq:DDI} and hDDI~\eqref{eq:DI} are partial differential inequalities and hard to solve. Only in the linear case, where the storage function can be chosen as a quadratic form, they boil down to a Riccati inequality~\cite{Schaft2017}. Instead of searching for a storage function $V_\mathrm{h}$ leading to the $\lph$-gain $\gamma_{\mathrm{h}p}$, we proceed as follows. 
	Since hDDI~\eqref{eq:DI} needs to hold $\forall\,u\in\lph$, it necessarily holds for $u\equiv 0$. Given the asymptotic stability of \eqref{eq:hnod_z}, a natural candidate storage function is any homogeneous Lyapunov function $V_\mathrm{l}$ of appropriate degree, who's existence is ensured by the converse Lyapunov theorem ~\cite{Bacciotti2005,Bhat2005}. Using Lemma~\ref{lem:homog_domination}, an appropriate scaling $a>0$ can be found that qualifies $V_\mathrm{h}= aV_\mathrm{l}$ as storage function. This makes it possible to calculate an upper bound $\hat{\gamma}_{\mathrm{h}p} \geq \gamma_{\mathrm{h}p}$.

	\subsection{Homogeneous Lyapunov and Storage Functions}

	Consider the scaled error dynamics~\eqref{eq:h1od_z} with zero input, i.e. the noise- and disturbance-free case ($\nu\equiv 0$ and $\delta \equiv 0$).  A Lyapunov function of homogeneity degree $d_\mathrm{V} = 2-d$ is given by~\cite{CruzZavala2016,CruzZavala2019} 
	\begin{equation}\label{eq:LF}
		V_\mathrm{l}(z_1,z_2) = \frac{1-d}{2-d}|z_1|^{\frac{2-d}{1-d}} - z_1 z_2 + \frac{1+\beta}{2-d}|z_2|^{2-d},\quad \beta > 0.
	\end{equation}
	It is positive definite and continuously differentiable for $d\in[-1,1)$ with homogeneous derivative  
	\begin{multline}\label{eq:dVl}
		\dot{V}_\mathrm{l}(z_1,z_2) = -\tilde k_1 \underbrace{\left|	\signpow{z_1}^{\frac{1}{1-d}}-z_2 	\right|^2}_{\mu(z_1,z_2)}\\
		+\tilde{k}_2 \underbrace{ \left[(1+\beta)	\left(z_1-\signpow{z_2}^{1-d}\right) \signpow{z_1}^{\frac{1+d}{1-d}}
		-\beta|z_1|^{\frac{2}{1-d}}\right]}_{\eta(z_1,z_2)}.
	\end{multline}
	Note that for $z_2 = \signpow{z_1}^{\frac{1}{1-d}}$ it simplifies to $\dot{V}_\mathrm{l} = 	-\tilde{k}_2\beta|z_1|^{\frac{2}{1-d}}$ which is negative. Thus, with the help of Lemma~\ref{lem:homog_domination}, $\dot V_\mathrm{l}$ can be rendered negative definite by appropriate choice of $\tilde{k}_1$.  From~\eqref{eq:dVl}, we can derive the lower bound as
	\begin{equation}\label{eq:condition_dVl_negative}
		\frac{\tilde{k}_1}{\tilde{k}_2} > \max_{z\in\R^2} g(z_1,z_2) \quad \text{with}\quad g(z_1,z_2) =  \frac{\eta(z_1,z_2)}{\mu(z_1,z_2)}.
	\end{equation}
	The function $g$ is upper semicontinuous and homogeneous of degree zero. Hence, it achieves a maximum and the search can be restricted to the homogeneous unit sphere, i.e. to $\mathcal{S}_{r_z,q}$ with $q \geq 1$. Choosing the gains conform with Condition~\eqref{eq:condition_dVl_negative}  ensures $\dot V_{\mathrm{l}} < 0$ and implies stability of the error dynamics~\eqref{eq:h1od_z} for $\nu\equiv 0$ and $\delta \equiv 0$. Since this is necessary for the following derivations, we assume throughout the rest of the article:

	\begin{ass}[Stabilizing Gains]\label{ass:gains}
		The differentiator gains $\tilde{k}_1$ and $\tilde{k}_2$ are chosen such that $V_\mathrm{l}$ is a Lyapunov function for the error dynamics~\eqref{eq:h1od_z} with $\nu\equiv 0$ and $\delta \equiv 0$. That is, Condition~\eqref{eq:condition_dVl_negative} is satisfied. 
	\end{ass}

	\begin{remark}[Independence of $L$]
		Note that the requirement for stabilizing gains
%		\begin{equation*}
	$
			\frac{\tilde{k}_1}{\tilde{k}_2} = \frac{k_1^2}{k_2} = \frac{\alpha_1^2 L^2}{\alpha_2 L^2} = \frac{\alpha_1^2}{\alpha_2}
%		\end{equation*}
	$	is scaling-invariant. Therefore, the scaling with $L$ does not affect stability of the error dynamics~\eqref{eq:h1od_z}.
	\end{remark}

	As pointed out by Zhang~\cite{Zhang2022} and discussed earlier, a scaled version of Lyapunov function $V_\mathrm{l}$ can be utilized as a storage function $V = aV_\mathrm{l}$, $a>0$ for dissipation inequality~\eqref {eq:DI}. 
	We adapt the ideas to the present case.
	
	\begin{thm}[Storage function for differentiator~\eqref{eq:h1od_z}]\label{thm:SF}
		Under Assumption~\ref{ass:gains} and with scaled gains~\eqref{eq:gains_scaled}, the function $V_\mathrm{h}(z_1,z_2) = aV_\mathrm{l}(z_1,z_2)$ serves as a storage function for the error dynamics~\eqref{eq:h1od_z} w.r.t. the homogeneous $\lp[2]$ supply rate~\eqref{eq:l2h_supply_rate}, if  we choose $a = \tilde{a}L$ with
		\begin{subequations}\label{eq:thm_SF}
		\begin{align}
			 \tilde{a} &> M = \max_{z\in\mathcal{S}_{r_z,2}} m(z)\quad\text{and}\\
				m(z) &= \frac{\alpha_1^2|z_2|^2 }{  \alpha_1 \mu(z_1,z_2)
				- \frac{\alpha_2 }{\alpha_1}\eta(z_1,z_2)},\label{eq:thm_SF_m}
		\end{align}
	\end{subequations}
		where $\mu$ and $\eta$ are defined in~\eqref{eq:dVl} and $M$ is independent of gain scaling $L$.
	\end{thm}

	\begin{proof}
		We rewrite the dissipation inequality~\eqref{eq:DI} and make use of Lemma~\ref{lem:homog_domination} twice. 
		The derivative $\dot{V}$ reads
		\begin{align*}
			\dot{V}(z)\! &= a\frac{\partial V_\mathrm{l}(z)}{\partial z}\dot z \\
			\!&=a\left[   - \tilde{k}_1 \left(\signpow{z_1}^{\frac{1}{1-d}} - z_2\right) \left(\signpow{z_1+\nu}^{\frac{1}{1-d}} - z_2\right)\right.\\ 
			\!& \left.  + \left(- z_1 + (1 + \beta)\signpow{z_2}^{1-d} \right) \left(-\tilde{k}_2 \signpow{z_1+\nu}^{\frac{1+d}{1-d}} + \tilde{\delta}\right)\right].
		\end{align*}
		With homogeneous 2-norm 
%		\begin{equation*}
			$\qhnorm{z}{r_z}{2}{2} =  |z_1|^{\frac{2}{1-d}} + |z_2|^{\frac{2}{1}} $,
%		\end{equation*}
		the respective output and input norms are
		\begin{equation*}
			\qhnorm{y}{r_y}{2}{2} = \tilde k_1^2|z_2|^2,\quad
			\qhnorm{u}{r_u}{2}{2} = |\nu|^{\frac{2}{1-d}} + |\delta|^{\frac{2}{1+d}}.
		\end{equation*} 
		We substitute the scaled gains~\eqref{eq:gains_scaled} and choose $a = \tilde{a} L$. A division by $L^2>0$ yields the dissipation inequality
		\begin{equation*}
			\hat{\mathcal{J}}(z,\nu,\delta) < 0,
		\end{equation*}
		where		
		\begin{multline}\label{eq:Jhat}
				\hat{\mathcal{J}}(z,\nu,\delta) = \tilde{a}  \Bigg[   - \alpha_1 \left(\signpow{z_1}^{\frac{1}{1-d}} - z_2\right) \left(\signpow{z_1+\nu}^{\frac{1}{1-d}} - z_2\right)\\
			+ \left(- z_1 + (1 + \beta)\signpow{z_2}^{1-d} \right) \left(-\frac{\alpha_2 }{\alpha_1}\signpow{z_1+\nu}^{\frac{1+d}{1-d}} + \frac{\delta}{L^2\alpha_1}\right)\Bigg]\\
			+ \alpha_1^2|z_2|^2 - \left(\frac{\hat{\gamma}}{L}\right)^2\left(|\nu|^{\frac{2}{1-d}} + |\delta|^{\frac{2}{1+d}}\right).
		\end{multline}	
		with a homogeneous of degree  $d_\mathcal{J}=2$ left-hand side.  \\
		To apply Lemma~\ref{lem:homog_domination}, observe that the previous inequality simplifies to
		\begin{equation*}
			-\tilde{a}\left[ \alpha_1 \mu(z_1,z_2) - \frac{\alpha_2 }{\alpha_1}\eta(z_1,z_2)\right]+ \alpha_1^2|z_2|^2 < 0
		\end{equation*}
		for $u\equiv 0$, where $\mu$ and $\eta$ are defined in~\eqref{eq:dVl}. It is satisfied if we choose $\tilde{a}$ as proposed in~\eqref{eq:thm_SF}.  With Assumption~\ref{ass:gains}, we know that $\dot V_\mathrm{l}$ is negative definite leading to a positive definite denominator of $m$ in~\eqref{eq:thm_SF_m}. Thus, $m$ is continuous $\forall z\neq 0$. Since $m$  is homogeneous of degree zero, it achieves a maximum that can be found on the unit sphere. \\
		Observe that $\qhnorm{u}{r_u}{2}{2}$ is non-negative. Given a proper choice of $\tilde{a}$, we can thus use Lemma~\ref{lem:homog_domination} to ensure the existence of $\hat{\gamma}^*$ such that dissipation inequality~\eqref{eq:DI} holds for every $\hat{\gamma} \geq \hat{\gamma}^*$. This qualifies $V$ as a storage function for the error dynamics~\eqref{eq:h1od_z} w.r.t. $\lph[2]$ supply rate~\eqref{eq:l2h_supply_rate} ($p=2$).		
	\end{proof}
	
	This directly leads to the following assumption.
	\begin{ass}[Storage function]\label{ass:SF}
		The storage function $V$ for the differentiator's error dynamics~\eqref{eq:h1od_z} reads \begin{equation*}
			V_\mathrm{h}(z_1,z_2) = \tilde a L V_\mathrm{l}(z_1,z_2),
		\end{equation*}
		where the constant $\tilde{a}>0$ is chosen conform with Theorem~\ref{thm:SF}.
	\end{ass}
	
%	\newpage
	\section{Homogeneous $\lp[2]$-Gain of the Differentiator}\label{sec:l2h-gain}
 	With the help of storage function $V_\mathrm{h}$ we estimate an upper bound of the $\lph[2]$-gain.  Then, we show that there exists a global minimum of the estimate w.r.t. gain-scaling $L$ and propose a procedure to obtain it.
 	
	\subsection{Estimation with Fixed Parameters}
	We use $V_\mathrm{h}$ to estimate an upper bound~$\hat{\gamma}$ on the homogeneous $\lp[2]$-gain from input~$u$ to output~$y$ for fixed values of the gains $\alpha_i$, $i=1,2$, gain scaling $L$, Lyapunov function parameter $\beta$ and Lyapunov function scaling $\tilde{a}$.
	
	Zhang~\cite{Zhang2022} proposes two approaches for this purpose depending on the system structure. For dynamics affine in the input, the worst-case input  $u^*$ can be calculated from the partial derivative of value function $\mathcal{J}$ in~\eqref{eq:DI_value_function}  w.r.t $u$ (in line with the linear case and classical $\mathcal{H}_\infty$-norm calculation). Then, the minimum $\hat{\gamma}$ is found such that $\mathcal{J}<0$ holds. Otherwise, $\hat{\gamma}^2$ is obtained by maximizing the remainder of dissipation inequality~\eqref{eq:DI} solved for $\hat{\gamma}^2$.
	
	In the present case~\eqref{eq:h1od_z}, the error dynamics are only affine in the disturbance $\delta$. This does not apply to the measurement noise $\nu$ (unless $d= 0$). Thus the result presented here adapts the ideas of Zhang, where we introduce the parameter $\gamma$ as argument of the respective functions to highlight their dependency.
	
	\begin{prop}[Estimation of the ${\lph[2]}$-gain]\label{prop:calculatioin}
		Under Assumptions~\ref{ass:gains} and~\ref{ass:SF}, an upper estimate $\hat{\gamma}_{\mathrm{h}2}=\hat\gamma(V_\mathrm{h}) $ on the  homogeneous $\lp[2]$-gain of the differentiator's error dynamics~\eqref{eq:h1od_z} can be calculated as
		\begin{equation*}
			\hat\gamma(V_\mathrm{h}) = \arg\min_{\gamma\geq 0}\left\lbrace\max_{\qhnorm{(z,~\nu)}{(r_z,r_\nu)}{2}{}=1} \tilde{ \mathcal{J}}(z,\nu;\gamma)< 0\right\rbrace
		\end{equation*}
		where
		\begin{align*}
%			&\tilde{\mathcal{J}}(z,\nu;\gamma) 			= 		-\tilde{a}   \alpha_1 \left(\signpow{z_1}^{\frac{1}{1-d}} - z_2\right) \left(\signpow{z_1+\nu}^{\frac{1}{1-d}} - z_2\right)\\
%			&+ \tilde a \left(- z_1 + (1 + \beta)\signpow{z_2}^{1-d} \right) \left(-\frac{\alpha_2 }{\alpha_1}\signpow{z_1+\nu}^{\frac{1+d}{1-d}} + \frac{1}{L^{2}\alpha_1}\signpow{ \frac{\tilde a (1+d)}{2 \alpha_1 \gamma^2}  \left(- z_1 + (1 + \beta)\signpow{z_2}^{1-d} \right) }^{\frac{1+d}{1-d}} \right)\\
%			&+\alpha_1^2|z_2|^2 - \left(\frac{\gamma}{ L}\right)^2\left(|\nu|^{\frac{2}{1-d}} + \left| \frac{\tilde a (1+d)}{2 \alpha_1 \gamma^2}  \left(- z_1 + (1 + \beta)\signpow{z_2}^{1-d} \right) \right|^{\frac{2}{1-d}}\right)
			\tilde{\mathcal{J}}(z,\nu;\gamma) 			&=	\left. \hat{\mathcal{J}}(z,\nu,\delta)\right|_{\delta = \delta^*}\quad\text{see~\eqref{eq:Jhat}}
			,\\
			\delta^* &=  \signpow{ \frac{\tilde a (1+d)}{2 \alpha_1 \gamma^2}  \left(- z_1 + (1 + \beta)\signpow{z_2}^{1-d} \right) }^{\frac{1+d}{1-d}}.
%			\tilde{\mathcal{J}}_\nu(z,\nu;\gamma) &= -\tilde{a}k_{1}\frac{1}{1-d}\left(\left\lceil z_{1}\right\rfloor ^{\frac{1}{1-d}}-z_{2}\right)\left|z_{1}+\nu \right|^{\frac{d}{1-d}}\\
%			& -\tilde{a}\frac{k_{2}}{k_{1}}\frac{1+d}{1-d}\left(-z_{1}+\left(1+\beta\right)\left\lceil z_{2}\right\rfloor ^{1-d}\right)\left|z_{1}+\nu \right|^{\frac{2d}{1-d}}\\
%			&-\frac{2}{1-d}\left(\frac{\gamma}{L}\right)^{2}\left\lceil \nu \right\rfloor ^{\frac{1+d}{1-d}}
		\end{align*}
	\end{prop}
	
	\begin{proof}
	In order to find the smallest $\gamma$ such that dissipation inequality~\eqref{eq:DI} holds for the given storage function~$V_\mathrm{h}$, we observe that the variable $\delta$ can be eliminated from $\hat{\mathcal{J}}$ by finding the maximum of $\hat{\mathcal{J}}$ w.r.t. $\delta$. Since
	\begin{multline*}
		\frac{\partial \hat{\mathcal{J}}(z,\nu,\delta) }{\partial \delta} = \frac{\tilde a}{L^2\alpha_1} \left(- z_1 + (1 + \beta)\signpow{z_2}^{1-d} \right) \\
		- \left(\frac{\gamma}{ L}\right)^2 \frac{2}{1+d} \signpow{\delta}^{\frac{1-d}{1+d}}
	\end{multline*} 
	is continuously differentiable in $\delta$ for $d < 1$ and
	\begin{equation*}
		\frac{\partial^2 \hat{\mathcal{J}}(z,\nu,\delta) }{\partial \delta^2} 
		= -\left(\frac{\gamma}{ L}\right)^2 \frac{2(1-d)}{(1+d)^2} |\delta|^{\frac{-2d}{1+d}} 
		< 0 \quad \forall \delta \neq 0, %d \in (-1,0],
	\end{equation*}
	it achieves a maximum w.r.t. $\delta$ at $\delta^*$ given in Proposition~\ref{prop:calculatioin}.
	
%	Aiming at a similar result for noise $\nu$, calculate
%	\begin{equation*}
%		\tilde{\mathcal{J}}_\nu(z,\nu;\gamma) = 	\frac{\partial \hat{\mathcal{J}}(z,\nu,\delta) }{\partial \nu}
%	\end{equation*}
%	and observe that $\tilde{\mathcal{J}}_\nu(z,\nu;\gamma) = 0$ is not easily solvable for $\nu$. It is therefore kept as a condition. Note that, for $d\in(-1,~0]$, the partial derivative $\tilde{\mathcal{J}}_\nu$ only exists whenever $\nu\neq -z_1$. However, consider
%	\begin{multline*}
%		\left.\hat{\mathcal{J}}(z,\nu,\delta)\right|_{\nu = -z_1} = \tilde{a}  \Bigg[   - \alpha_1 \left(\signpow{z_1}^{\frac{1}{1-d}} - z_2\right) \left(- z_2\right)\\
%		+ \left(- z_1 + (1 + \beta)\signpow{z_2}^{1-d} \right) \left( + \frac{\delta}{L^2\alpha_1}\right)\Bigg]\\
%		+ \alpha_1^2|z_2|^2 - \left(\frac{\gamma}{L^{2}}\right)^2\left(|\nu|^{\frac{2}{1-d}} + |\delta|^{\frac{2}{1+d}}\right)
%	\end{multline*}
%	\todo[inline]{show that maximum is not achieved on $\nu = -z_1$}
%	Thus, the maximum is achieved at a differentiable point.
%	
	Since $\tilde{ \mathcal{J}}$ is homogeneous, it is sufficient to restrict the search on the homogeneous unit sphere.

	\end{proof}
	
	In the two-dimensional case discussed here, finding a maximum of $\tilde{ \mathcal{J}}$ is a problem in three variables $z_1,z_2,\nu$. By restricting us to the (homogeneous) unit sphere, we gain the reduction of one dimension and a restricted search area from an unbounded to a compact set. Still, we cannot show convexity for $d\neq 0$ and thus require conscientious search on the entire sphere (see e.g. Fig.~\ref{fig:find_gamma} for $d = -0.5$ and a specific parameter set). Analog to the calculation of the $\lp[2]$-gain based on the Hamiltonian matrix in the linear case, a sub-optimal $\hat{\gamma}(V_\mathrm{h})$ can be found using e.g. a bisection algorithm. Instead of checking the location of the Hamiltonian's eigenvalues, we are interested in the sign of $\max \tilde{ \mathcal{J}}$.
	
	\subsection{Minimization of the estimated homogeneous $\lp[2]$-Gain}
	Being able to calculate an estimate $\hat\gamma$ of the homogeneous $\lp[2]$-gain, we make use of gain scaling $L$ to minimize $\hat{\gamma}$. That is, consider the scaled gains~\eqref{eq:gains_scaled} with fixed $\alpha_1,\alpha_2$ but variable scaling $L$.

	\begin{thm}[Existence of optimal gain-scaling $L^*$]\label{thm:minL}
		Under Assumptions~\ref{ass:gains} and~\ref{ass:SF}, there exists an optimal scaling $L^*>0$ that (globally) minimizes the estimate $\hat{\gamma}_{\mathrm{h}2} = \hat{\gamma}(V_\mathrm{h},L)$, 
	\end{thm}
	
	\begin{proof}
	To analyze the effect of $L$ on the homogeneous $\lp[2]$-gain, i.e. the minimum value of $\hat \gamma$ such that $\hat{\mathcal{J}}<0$ holds, consider the two limit cases  $\delta\equiv 0$ and $\nu\equiv 0$, respectively.
	
  	In case of no disturbance, i.e. $\delta\equiv 0$, the inequality with $\hat{\mathcal{J}}$ of~\eqref{eq:Jhat} simplifies to
	\begin{multline*}
		\hat{\mathcal{J}}(z,\nu,0) = \tilde{a}  \left[   - \alpha_1 \left(\signpow{z_1}^{\frac{1}{1-d}} - z_2\right) \left(\signpow{z_1+\nu}^{\frac{1}{1-d}} - z_2\right)\right.\\
		\left.+ \left(- z_1 + (1 + \beta)\signpow{z_2}^{1-d} \right) \left(-\frac{\alpha_2 }{\alpha_1}\signpow{z_1+\nu}^{\frac{1+d}{1-d}} \right)\right]\\
		+\alpha_1^2|z_2|^2 - \left(\frac{\hat\gamma}{ L}\right)^2 |\nu|^{\frac{2}{1-d}}  < 0.
	\end{multline*}	
	Observe that the $\lph[2]$-gain estimate is $\hat\gamma = \tilde{\gamma} L$, where $\tilde{\gamma}$ is independent of $L$.
	
	In contrast, with zero noise $\nu\equiv 0$, it reads
	\begin{multline*}
		\hat{\mathcal{J}}(z,0,\delta) = \tilde{a}  \left[   - \alpha_1 \left(\signpow{z_1}^{\frac{1}{1-d}} - z_2\right) \left(\signpow{z_1}^{\frac{1}{1-d}} - z_2\right)\right.\\
		\left.+ \left(- z_1 + (1 + \beta)\signpow{z_2}^{1-d} \right) \left(-\frac{\alpha_2 }{\alpha_1}\signpow{z_1}^{\frac{1+d}{1-d}} + \frac{1}{\alpha_1}\bar\delta\right)\right]\\
		+\alpha_1^2|z_2|^2 - \left(\hat\gamma L^{\frac{1-d}{1+d}}\right)^2\left|\bar\delta\right|^{\frac{2}{1+d}} <0 
	\end{multline*}
	with $\bar{\delta} = \frac{\delta}{L^2}$. Note that $\frac{1-d}{1+d}>0$ for $d<1$. This yields $\hat\gamma = \frac{\tilde{\gamma}}{L^{\frac{1-d}{1+d}}}$ with $L$-independent $\tilde{\gamma}$.

	Thus, if noise $\nu$ is present, increasing $L$ leads to a larger estimate $\hat\gamma$ of the $\lph[2]$-gain and the contrary holds for a nonzero disturbance $\delta$. 
	\end{proof}
	
	We therefore try to find an optimal gain scaling $L^*$ that minimizes the estimated $\lph[2]$-gain from measurement noise $\nu$ and disturbance $\delta$ to the differentiation error $y$ of the homogeneous differentiator.
	To find $L^*$, several strategies are possible that all build upon Proposition~\ref{prop:calculatioin}.
	Implementation-wise easily, a local optimum can be obtained with the bisection algorithm. However, more sophisticated search strategies are possible to reduce the computational effort. 
	Although we show the existence of a global optimum, we do not show uniqueness. Therefore, several local minima are possible in principle. However, observing convexity in Fig.~\ref{fig:gamma_L}, we expect the minimum to be global.

	\subsection{Numerical Evaluation}\label{sec:numerics}
	To illustrate the calculation and minimization of the $\lph[2]$-gain estimate $\hat{\gamma}=\hat{\gamma}_{\mathrm{h}2}$ for error dynamics~\eqref{eq:h1od_z}, we go through the proposed procedure once and provide a simple, yet illustrative simulation example. 
	
	Consider the values of $\alpha_1$, $\alpha_2$ and $\beta$ in Table~\ref{tab:parameters} and find that Assumptions~\ref{ass:gains} and~\ref{ass:SF} are satisfied.\footnote{Since the maximum of functions $g$ and $m$ depend on the homogeneity degree $d$, their exact values are omitted here for the general case. However, Table~\ref{tab:parameters} holds explicit values for the simulation example.}
	For a variety of homogeneity degrees $d\in (-1,0]$ and gain-scalings $L\in[0.3,2]$, a bisection algorithm is used to estimate the respective values~$\hat{\gamma}$ of the upper bound on the homogeneous $\lp[2]$-gain (see Proposition~\ref{prop:calculatioin}). The results are presented in Fig.~\ref{fig:gamma_L}, where the continuous lines correspond to the estimated values and the black dashed line represents the actual $\lp[2]$-gain for the linear case (i.e. the $\mathcal{H}_\infty$-norm). 
	\begin{figure}[ht]
		\centering
		\includegraphics[width=1.0\linewidth]{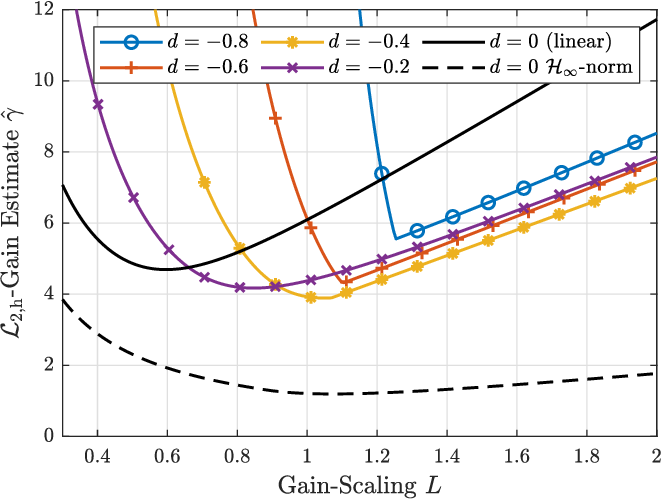}
		\caption{Estimated homogeneous $\lp[2]$-gain with fixed parameters (Tab.~\ref{tab:parameters}) for a variety of scalings $L$ and homogeneity degrees~$d$, where $d=0$ corresponds to the linear case.}
		\label{fig:gamma_L}
%			\vspace*{-.3cm}
	\end{figure}
	A comparison is only possible for the linear case, since the exact $\lph[2]$-gain cannot be calculated for $d\neq 0$, yet. We observe that the estimate~$\hat{\gamma}$ is of the same order and similar shape as the actual $\lp[2]$-gain, i.e. there exists an optimal gain-scaling~$L^*$ which minimizes the effect of noise and disturbance on the differentiation error.  Thus, we expect to get reasonable results for $d\neq 0$ as well. With decreasing homogeneity degree (given the parameters of Table~\ref{tab:parameters}), the estimated $\lph[2]$-gain increases for small values of $L$. This means, that the estimated effect of disturbance~$\delta$ becomes more severe if gains are small (see proof of Theorem~\ref{thm:SF}) and  input-to-state stability would be lost considering the limit case $d\to -1$. On the other hand, the estimated effect of measurement noise does not necessarily increase with decreasing degree $d$. In every case, a minimum w.r.t. $L$ can be observed. 

	Note that these results depend on $V_\mathrm{h}$, i.e. on both $\beta$ and $\tilde a$. Parameter variation results in a shifted optimum in terms of $L$ and $\hat\gamma$. However, the overall shape is not affected.

		\begin{table}[h]
		\caption{Parameters used in the simulation example.}\label{tab:parameters}
		\centering
		\begin{tabular}{l|c|c|c|c|c|c}
			\textbf{Parameter}	 &	$\alpha_1$	&	$\alpha_2$		&$ \beta$ &	$d$	&$M$			&$\tilde{a}$\\\hline
			\textbf{Value}				&	$3$					&	$1.5\sqrt{3}$ 	& $1$			&  $-0.5$			&$15.72$	&$16.72$
		\end{tabular}
	\end{table}
	
	In order to illustrate the results of a $\lph[2]$-gain minimization with respect to $L$, we focus on one exemplary homogeneity degree of $d=-0.5$ and the remaining parameters summarized in Table~\ref{tab:parameters}.
	Indeed, Assumption~\ref{ass:gains} is satisfied, since (see~\eqref{eq:condition_dVl_negative}) 
%	\begin{equation*}
		$\displaystyle \max_{\qhnorm{z}{r_z}{2}{}=1} g(z_1,z_2) = 1.389 ~<~ \frac{\alpha_1^2}{\alpha_2} = 3.464$, 
%	\end{equation*}
	i.e. $V_\mathrm{l}$ is a Lyapunov function for~\eqref{eq:h1od_z} with $u = (\nu,~\delta)^\top\equiv 0$.
	Furthermore, Assumption~\ref{ass:SF} holds with $\tilde{a}>M$ of~\eqref{thm:SF}. Therefore, $V_\mathrm{h}=\tilde{a}LV_\mathrm{l}$ is used to estimate $\hat{\gamma}$. With a bisection algorithm, the optimal gain-scaling of $L^* = 0.975$ is found which leads to the smallest $\lph[2]$-gain estimate of $\hat{\gamma}^* = 3.990$.  
	The strictly negative value function $ \tilde{ \mathcal{J}}(z,\nu;\hat\gamma^*,L^*)$ for $(z,\nu)$ on the homogeneous unit sphere is provided in Fig.~\ref{fig:find_gamma}. It shows to be non-convex and makes conscientious search necessary in order to find $\hat{\gamma}$.
		\begin{figure}[ht]
		\centering
		\includegraphics[width=1.0\linewidth]{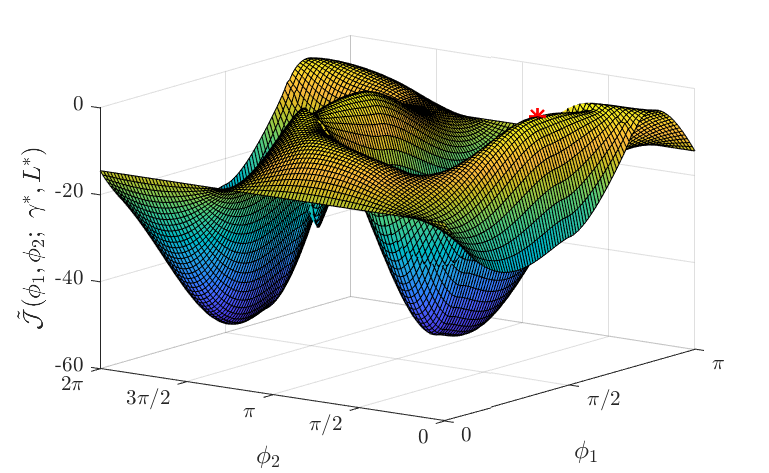}
		\caption{Value function $ \tilde{ \mathcal{J}}(\phi_1,\phi_2;\hat\gamma^*,L^*) =  \left.\tilde{ \mathcal{J}}(z,\nu;\hat\gamma^*,L^*)\right|_{\qhnorm{(z,~\nu)}{(r_z,r_\nu)}{2}{}=1}$ evaluated on the homogeneous unit sphere with the parameters of Table~\ref{tab:parameters} shows to be strictly negative for $\hat{\gamma}^*$ and $L^*$ with a maximum marked by a red asterisk.}
		\label{fig:find_gamma}
		%			\vspace*{-.3cm}
	\end{figure}

	We investigate the actual effect of noise and disturbance on the differentiation error by means of an illustrative simulation example. For this purpose,  the noisy signal
	\begin{equation*}
		f_\mathrm{n}(t) = f_0(t) + \nu(t)
				= a_0\sin(\omega_0 t) + a_\nu \sin(\omega_\nu t),
	\end{equation*} 
	is differentiated once,
	where $a_0 = 0.5$, $\omega_0 = 0.5$ and $a_\nu = 0.002$, $\omega_\nu=1000$. From the error dynamics~\eqref{eq:h1od_z}, it can be seen that disturbance~$\delta$ results in $\delta(t) = -a_0 \omega_0^2 \sin(\omega_0 t)$.
	To simulate measurement of the noisy signal, the initial conditions of the error dynamics are chosen as $z_1(0) = 0$ and $z_2(0) = 0.02$ resulting in nonzero initial differentiation error $y(0)=\alpha_1 L\, z_2(0)$.
	We simulate ten periods of the base signal $f_0$, i.e. $t\in[0,T]$ with $T = 10\frac{2\pi}{\omega_0}\approx 125.7\,\mathrm{s}$ and fixed sample-time of $\tau_\mathrm{s} = 1\times 10^{-4}\,\mathrm{s}$ using Euler's method. For clarity reasons, only two periods are depicted in Fig.~\ref{fig:sim_y}.
	\begin{figure}[ht]
		\centering
		\includegraphics[width=1.0\linewidth]{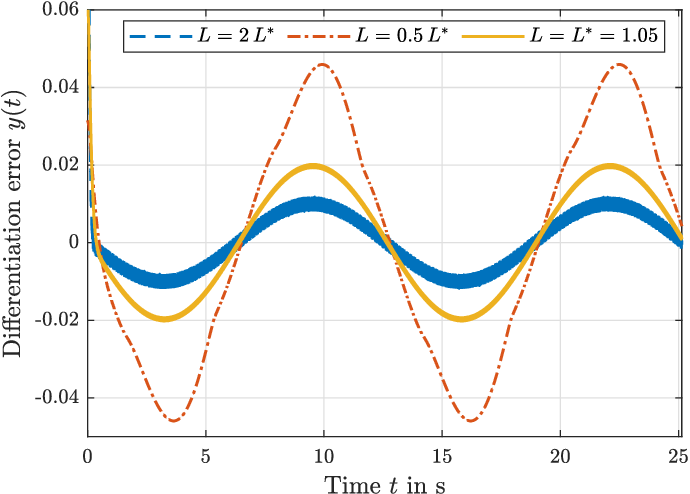}
		\caption{First two periods of differentiation error $y(t)$ for gain-scaling $L\in\{2L^*,0.5L^*,L^*\}$ when differentiating~$f(t)$.}
		\label{fig:sim_y}
		%			\vspace*{-.3cm}
	\end{figure}
	
	Apparently, $L=L^*$ yields a reasonable trade-off between the effect of measurement noise $\nu$ and disturbance $\delta$, which is in accordance with the theoretical observations of Theorem~\ref{thm:minL}.
	
	Finally, let us illustrate the logics of the homogeneous $\lp[2]$-gain compared to the standard (local) $\lp[2]$-gain discussed in Section~\ref{sec:motivation} with the present choice of parameters and signals. 
	Consider the specific quotients 
	\begin{equation*}
		\Gamma_\mathrm{T}(y,u) \!=\! \frac{\| y\|_{\mathcal{L}_{2,\mathrm{T}}}}{\| u\|_{\mathcal{L}_{2,\mathrm{T}}}} \!=\! 0.16,~ 
		\Gamma_{\mathrm{h,T}}(y,u) \!=\! \frac{\| y\|_{\mathcal{L}_{2\mathrm{h,T}}}}{\| u\|_{\mathcal{L}_{2\mathrm{h,T}}}} \!=\! 0.92
	\end{equation*}
	for $L=L^*$, where $u = (\nu,~\delta)^\top$ and $\mathcal{L}_{2,\mathrm{T}}$ and $\mathcal{L}_{2\mathrm{h,T}}$ are the respective truncated norms (i.e. $t\in[0,T]$).
	Applying the dilated input $\tilde{u}(\tilde t) = \Delta_\kappa^{r_u}(u(\kappa^{-r_t}t))$ leads to $\tilde{y}(\tilde t)=\Delta_\kappa^{r_y}(y(\kappa^{-r_t}t))$ with $\kappa = 2$ and the weights defined in Section~\ref{sec:h1od_z}. The resulting quotients read
	\begin{equation*}
			\Gamma_\mathrm{T}(\tilde y,\tilde u) =  0.23\neq \Gamma(y,u)_\mathrm{T},~%\text{and}\quad
			\Gamma_{\mathrm{h,T}}(\tilde y,\tilde u) =0.92= \Gamma_{\mathrm{h,T}}(y,u)
	\end{equation*}
	which exemplary shows that the standard $\lp[2]$-gain is not invariant under homogeneous dilation.
	 
	 As derived in the motivation example, considering $\delta\equiv 0$  and $\nu\equiv 0$ leads to the respective quotients
	 \begin{equation*}
	 	\frac{	\Gamma(\tilde y,\tilde \nu) }{\Gamma(y,\nu)} = %\kappa^{r_y-r_nu} = 
	 	\kappa^{d},\quad\kappa>0
	 	\quad\text{and}\quad 
 		\frac{	\Gamma(\tilde y,\tilde \delta) }{\Gamma(y,\delta)} = %\kappa^{r_y-r_\delta} =
 		\kappa^{-d}
 		.
	 \end{equation*}
	 With $d = -0.5$ and the parameters of the simulation example, the corresponding graphs are presented in Fig.~\ref{fig:gamma_quotient} and illustrate that only the homogeneous $\lp[2]$-gain is constant under homogeneous dilation, since it is homogeneous of degree zero.
	 This underlines the purpose of utilizing the homogeneous $\lp[2]$-gain for homogeneous systems.
	 	\begin{figure}[ht]
	 	\centering
	 	\includegraphics[width=1.0\linewidth]{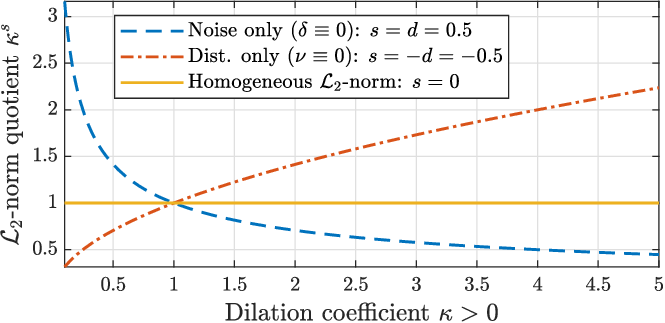}
	 	\caption{Quotients of the $\lp[2]$-norms of the dilated and original input- and output-signals for the special cases $\delta\equiv 0$ and $\nu\equiv 0$, respectively, and quotient of the $\lph[2]$-norms. The parameters are consistent with the preceding example.}
	 	\label{fig:gamma_quotient}
	 	%			\vspace*{-.3cm}
	 \end{figure}
%	\begin{figure}[ht]
%		\centering
		%	\includegraphics[width=1.0\linewidth]{Pendel.png}
%		\def\svgwidth{.95\linewidth}
%		\input{pendel_tex.pdf_tex}
%		\caption{Schematic representation of the pendulum on a cart.}
%		\label{fig:cart-pendulum}
%		\vspace*{-.3cm}
%	\end{figure}

%	\begin{remark}
%	
%		\end{remark}
	
%	\begin{figure}[ht]
%		\centering
%		\includegraphics[width=1.0\linewidth]{Syu.eps}
%		\caption{Input disturbance sensitivity function $S_\mathrm{u}(s)$.}
%		\label{fig:Syu}
%		\vspace*{-.3cm}
%	\end{figure}

	\section{Conclusions}
	\label{sec:conclusions}
	We have shown that there exists an optimal gain-scaling $L^*$ which locally minimizes the homogeneous $\lp[2]$-gain estimate of the differentiation error dynamics of dimension $n=2$ and homogeneity degree $d\in(-1,1)$.  
	The estimation is based on dissipativity of the error dynamics with respect to the homogeneous $\lp[2]$ supply rate, where a scaled Lyapunov function is utilized as storage function.  The numerical evaluation for an exemplary parameter set underlines that reasonable results can be achieved with the proposed approach. Furthermore,  an optimization w.r.t. the homogeneity degree $d$ seems possible. 
	Since the estimate depends on the storage function, future work covers a larger family of storage functions, where the parameters $\tilde{a}$, $\beta$ and the remaining degree of freedom are utilized to render the estimate less conservative. Moreover, the generalization towards the arbitrary-order homogeneous differentiator will be considered. Note that the theoretical framework introduced is valid for the general case, but the numerical calculation of the estimated homogeneous gain and the optimization are more complex.
	
%	\newpage
		\printbibliography
%	\bibliography{literature}
\end{document}